\documentclass[twocolumn, pra , footinbib 
]{revtex4-1}
\usepackage{graphicx}
\usepackage{amsfonts}
\usepackage{amssymb}
\usepackage{amsthm}
\usepackage{array}
\usepackage{amsmath}
\usepackage{verbatim} 
\usepackage{hyperref}
\usepackage{color}
\usepackage{bbold}
\usepackage{epstopdf}
\usepackage{mathtools}
\usepackage{enumerate}
\usepackage{subcaption}
\usepackage{physics}
\usepackage{subcaption}
\usepackage{soul,xcolor}
\usepackage{newtxtext,newtxmath}
\usepackage{hyphenat}

\newtheorem{proposition}{Proposition}
\newtheorem*{proposition*}{Proposition}

\newtheorem{theorem}{Theorem}
\newtheorem*{theorem*}{Theorem}

\newtheorem*{corollary*}{Corollary}

\theoremstyle{definition}
\newtheorem*{remark}{Remark}

\def\Tr{\mathrm{Tr}}
\setstcolor{red}

\newcommand{\free}{\mathcal{F}}
\newcommand{\freeop}{\mathcal{O}}
\newcommand{\R}{R}
\let\succ\relax
\newcommand{\succ}{\mathrm{succ}}
\newcommand{\proj}[1]{\ket{#1}\!\bra{#1}}

\begin{document}
\title{Fisher information universally identifies quantum resources}
\author{Kok Chuan Tan}
\email{bbtankc@gmail.com}
\affiliation{ School of Physical and Mathematical Sciences, Nanyang Technological University, Singapore 637371, Republic of Singapore}
\author{Varun Narasimhachar}
\affiliation{ School of Physical and Mathematical Sciences, Nanyang Technological University, Singapore 637371, Republic of Singapore}
\author{Bartosz Regula}
\affiliation{ School of Physical and Mathematical Sciences, Nanyang Technological University, Singapore 637371, Republic of Singapore}

\begin{abstract}
We show that both the classical as well as the quantum definitions of the Fisher information faithfully identify resourceful quantum states in general quantum resource theories, in the sense that they can always distinguish between states with and without a given resource. This shows that all quantum resources confer an advantage in metrology, and establishes the Fisher information as a universal tool to probe the resourcefulness of quantum states. We provide bounds on the extent of this advantage, as well as a simple criterion to test whether different resources are useful for the estimation of unitarily encoded parameters. Finally, we extend the results to show that the Fisher information is also able to identify the dynamical resourcefulness of quantum operations.
\end{abstract}

\maketitle

\noindent{\it Introduction ---} The Fisher information (FI) plays an important foundational role in quantum information science. In quantum metrology and sensing, it determines the ultimate limits of precision of our measurement devices via the well known quantum Cram{\'e}r-Rao bound~\cite{Tan2019, Sidhu2020, Braunstein1994}. Existing applications include interferometry~\cite{Caves1981, Dowling2008, Schnabel2010}, magnetometry~\cite{Taylor2008, Bhattacharjee2020}, thermometry~\cite{Correa2015, DePasquale2016}, quantum illumination~\cite{Lloyd2008, Zhuang2017, Sanz2017}, displacement sensing~\cite{Yadin2018, Kwon2019}, among others. Crucially, such applications exploit the use of well-studied nonclassical quantum properties such as coherence~\cite{Baumgratz2014, Streltsov2017}, entanglement~\cite{Horodecki2009, Plenio2007}, and negative quasiprobabilities~\cite{Cahill1969, Tan2020} in order to demonstrate the intrinsic superiority of quantum measurement devices over classical ones. FI has also been used to study nonclassical features of quantum systems such as quantum coherence~\cite{Giorda2017, Tan2018-2} and entanglement~\cite{Pezze2009, Toth2012}.

Traditionally, different notions of nonclassicality in quantum mechanics have been studied independently. As such, the theoretical tools and quantities that were developed in the past typically probe a single nonclassical feature at a time. However, recent developments have made tremendous strides in providing a unified framework to study not only several disparate notions of nonclassicality~\cite{Tan2016, sperling_2015, Tan2018}, but also more general resources of quantum systems~\cite{horodecki_2012,Chitambar2019}.
This has led to the discovery of physical tasks and operational quantities that are relevant in not just one particular resource theory, but also in general settings. This motivates us to consider quantities that are universally applicable in the sense that they maintain a physically meaningful interpretation while being able to identify \textit{every} state which is considered ``nonclassical'' or ``resourceful'' within the physical constraints of the given resource theory~\cite{piani_2009,piani_2015,brandao_2015,takagi_2019-2,bae_2018,takagi_2019,uola_2019-1,uola_2020-1,ducuara_2020}. An example of such quantities would be the class of robustness measures~\cite{vidal_1999}, which are well-defined quantifiers of any quantum resource that find use in quantifying the operational advantages of resources in channel discrimination tasks~\cite{takagi_2019-2,takagi_2019,regula_2021}. Not all meaningful quantities can identify all resource states of interest --- for instance, in the theory of quantum entanglement, there exist important tasks such as distillation or quantum teleportation for which certain classes of entangled states are useless~\cite{horodecki_1998,Horodecki2009}, and so quantities based on such tasks fail to faithfully characterize entanglement as a resource. This raises the questions of which tasks and quantities can be considered as universal witnesses of general quantum resources, and in which settings different notions of quantumness can provide tangible practical advantages.

In this work, we show that the FI universally characterizes the resources of quantum states, regardless of the specific resource in consideration, in that it is able to identify every resourceful state in general quantum resource theories. An immediate implication of this result is that quantum resources are always useful for quantum metrology, since there always exists some metrological problem where a resourceful probe outperforms a probe which does not possess a given resource. This also implies that the FI can be used as a generic tool to probe the resources of any system. We then establish theoretical bounds on the advantages provided by a given quantum states by relating the FI to the robustness measure of the given resource. We additionally provide a simple criterion for testing whether a given resource is useful for unitary parameter estimation. Finally, we also show that the FI can identify the resources of quantum operations in general resource theories.\vspace{2ex}

\noindent{\it Preliminaries ---} We first define what a general quantum resource theory means in our context. Let $\mathcal{S}$ be the state space that describes a quantum system. A quantum resource theory is composed of a well-defined set of free quantum states $\mathcal{F}$ --- depending on the setting, these can be understood as the classical states, or the states without the given resource. Such states are accompanied by some set of free quantum operations $\mathcal{O}$. In order to account for as many possible formulations of quantum resources as possible, only a minimal set of assumptions are imposed on the sets $\mathcal{F}$ and $\mathcal{O}$. The set of free states $\mathcal{F}$ is assumed to be some closed and convex, but otherwise arbitrary subset of the state space $\mathcal{S}$. The physical interpretation of the assumption of convexity is that if one statistically mixes two free states together, the output will remain free. Given a well-defined set $\mathcal{F}$, we assume that $\mathcal{O}$ is some set of quantum maps that satisfies $\Phi(\sigma) \in \mathcal{F}$ if $\sigma \in \mathcal{F}$ and $\Phi\in\mathcal{O}$. These very minimal assumptions made on $\mathcal{F}$ and $\mathcal{O}$ maximizes the generality of the subsequent results.

We now define the FI. In a typical scenario, there are two types of FI considered in quantum information. The first type is the FI one obtains from the classical post\hyp processing of statistical data. This data is typically obtained from the measurement output of a fixed measurement setup. We will refer to this as the classical FI (CFI), and denote it as $F_C$.

The second type of FI is the maximum CFI one can obtain over all possible quantum measurements. This is typically called the quantum FI (QFI), and will here be denoted $F_Q$. By definition, we see that the CFI is a simple lower bound to the QFI.

Suppose we have a quantum channel $\Phi_\theta$ that depends on some real parameter $\theta$, and we would like to estimate $\theta$. In order to do this, we pass a state $\rho$, called a probe, through the quantum channel $\Phi_\theta$ and then perform a quantum measurement (positive operator--valued measure) $M = \{ M_i \}$, where $M_i$ are positive operators such that $\sum_i M_i = \openone$. This results in the measurement statistics $P(i\mid \theta) = \Tr[\Phi_\theta(\rho)M_i]$. The maximum information about $\theta$ that we can obtain from the statistics $P(i\mid \theta)$ is quantified by the CFI, which is given by~\cite{Kay1993, Lehmann1998}:
\begin{align}
F_C(\rho \mid \Phi_\theta, M)  \coloneqq \sum_i  P(i\mid \theta)  \left[\pdv{\log P(i\mid \theta)}{\theta}\right]^2 . \label{def::classFish}
\end{align}

In most cases, the quantum channel is fixed, so we can suppress the dependence on $\Phi_\theta$ and use the simplified notation $F_C(\rho \mid M)$ instead. A similar notation will also be employed for the QFI $F_Q$.

In general, the FI is to be evaluated with respect to some given value of $\theta$. Since the QFI $F_Q(\rho)$ is just the CFI maximized over all possible measurements $M$, the former depends only on the state $\rho$ and the quantum channel $\Phi_\theta$, and we have $F_C(\rho \mid M)  \leq  F_Q(\rho) $. Formally, the QFI is given by the expression
\begin{align}
F_Q(\rho) \coloneqq \Tr(\rho_\theta D^2_\theta),
\end{align} where $D_\theta$ is the symmetric logarithmic derivative~\cite{Helstrom1967, Helstrom1968}, satisfying the equation $\pdv{\theta} \rho_\theta = \acomm{\rho_\theta}{D_\theta}/2$.\vspace{2ex}

\noindent{\it Nonclassicality from the FI ---} We will now establish our main result, which is that the FI can reveal general quantum resources. In order to do this, we need to demonstrate that for any resourceful state $\rho \not{\in} \mathcal{F}$, there always exists a metrological problem represented by some quantum channel $\Phi_\theta$ together with some measurement $M$, where the resulting FI correctly identifies $\rho$ to possess some resources.

Suppose we would like to witness the resources of a state via the CFI. One way to go about doing this is to consider the quantity:
\begin{align} \label{def::NC}
N_C (\rho \mid M) \coloneqq F_C(\rho \mid M) - \max_{\sigma \in  \free}F_C(\sigma \mid M),
\end{align} 
where $F_C(\rho \mid M)$ is the CFI obtained by performing the measurement $M$ on the state $\Phi_\theta(\rho)$, and the maximization is over $ \free$, the set of free states in any given resource theory. $N_C$ then quantifies the minimum quantum advantage of a nonclassical state $\rho$ over all possible classical states for a given metrological problem.

We can also consider a similar definition using the QFI:

\begin{align} \label{def::NQ}
N_Q (\rho) \coloneqq F_Q(\rho ) - \max_{\sigma \in  \free} F_Q(\sigma).
\end{align} 

We see that if $N_Q(\rho)>0$ or $N_C(\rho \mid M) > 0$, then the FI that is obtained using $\rho$ exceeds that which can be obtained using any resourceless state $\sigma \in  \free$. Since the excess FI can only be attributed to the given resource, the state $\rho$ must be resourceful.

Upon first inspection, one may expect, since the QFI contains more information about the metrological utility of the quantum state than the CFI, that $N_Q$ performs better than $N_C$ at identifying nonclassicality. This is in fact incorrect. To see this, recall that the QFI is the CFI optimized over all possible quantum measurements. Suppose $M^\star$ is the optimal measurement. This implies that for any channel $\Phi_\theta$, it is always possible to find a measurement $M^\star$ such that $N_C (\rho \mid M^\star) \geq N_Q(\rho)$. In general, we therefore see that the gap between resourceful and free states is larger using the CFI compared to the QFI, i.e., the CFI is able to identify more states as resourceful. Indeed, there exist scenarios where $N_Q$ witnesses strictly fewer states than $N_C$. This is further discussed in the Supplemental Material~\cite{CompanionPaper}.

In the following Theorem, we show that both the classical and the quantum versions of the FI can be used to identify general quantum resources.

\begin{theorem} \label{thm::existence}
	There exists a parameter estimation problem with quantum channel $\Phi_\theta$ and measurement $M$ such that  $N_C(\rho \mid M) > 0$ and $N_Q(\rho) > 0$ if and only if $\rho \notin \free$.
\end{theorem}

\begin{proof}
	We will first prove the statement for $N_C$.
	
	It is immediately clear that if $N_C(\rho \mid M) > 0$ or $N_Q(\rho ) > 0$ for any parameter estimation problem, then $\rho$ must be resourceful, which proves the ``only if" direction.
	
	To prove the converse direction, we use a result from Ref.~\cite{takagi_2019-2}, which states that if $\rho$ is resourceful, then there exists a pair of quantum channels $\{A_0, A_1 \}$ and POVM $\{ \pi_0, \pi_1 \}$ such that $p_{\text{succ}}(\rho)>\max_{\sigma \in  \free}p_{\text{succ}}(\sigma)$, where $p_{\text{succ}}(\rho)  \coloneqq \frac{1}{2}\Tr[A_0(\rho)\pi_0]+\frac{1}{2}\Tr[A_1(\rho)\pi_1]$.
	
	Let $\{A_0, A_1 \}$ and POVMs $\{ \pi_0, \pi_1 \}$ be any such channel satisfying the above condition. We also introduce some state $\sigma_0 \neq \rho$ which will be specified later.
	
	We now consider the following series of quantum channels acting on some arbitary state $\tau$:
	\begin{align}
	\Lambda_1(\tau) &= \tau \otimes \frac{1}{2} \openone \otimes [\theta \ketbra{0}+(1-\theta) \ketbra{1}] 
	\end{align}
	
	\begin{align}
	\begin{split}\Lambda_2 \circ \Lambda_1(\tau) &= \tau\otimes \frac{1}{2} \openone \otimes \theta \ketbra{0} \\ & + \sigma_0 \otimes \frac{1}{2} \openone \otimes (1-\theta) \ketbra{1}\end{split} 
	\end{align}
	
	\begin{align}
	\begin{split}\Lambda_3 \circ \Lambda_2 \circ \Lambda_1(\tau) &= A_0(\tau) \otimes \frac{1}{2} \ketbra{0} \otimes \theta \ketbra{0} \\ 
	&+ A_1(\tau) \otimes \frac{1}{2} \ketbra{1} \otimes \theta \ketbra{0} \\
	&+ A_0(\sigma_0) \otimes \frac{1}{2} \ketbra{0} \otimes (1-\theta) \ketbra{1} \\ 
	&+ A_1(\sigma_0) \otimes \frac{1}{2} \ketbra{1} \otimes (1-\theta) \ketbra{1} \end{split} 
	\end{align}
	
	\begin{align}
	\begin{split}\Lambda_4 \circ \Lambda_3 &\circ \Lambda_2 \circ \Lambda_1(\tau) = \\
	&\{\Tr[\Lambda_3 \circ \Lambda_2 \circ \Lambda_1(\tau) \pi_0 \otimes \ketbra{0} \otimes \openone] \\&+ \Tr[\Lambda_3 \circ \Lambda_2 \circ \Lambda_1(\tau) \pi_1 \otimes \ketbra{1} \otimes \openone] \}\ketbra{0} \\
	&\{ \Tr[\Lambda_3 \circ \Lambda_2 \circ \Lambda_1(\tau) \pi_0 \otimes \ketbra{1} \otimes \openone] \\&+ \Tr[\Lambda_3 \circ \Lambda_2 \circ \Lambda_1(\tau) \pi_1 \otimes \ketbra{0} \otimes \openone] \}\ketbra{1} \end{split}	
	\end{align}
	
	Finally, we perform the projection $P_0 \coloneqq \ketbra{0}$ and $P_1 \coloneqq \ketbra{1}$ to obtain the statistics 
	\begin{align}
	P(0 \mid \theta) = \theta p_{\text{succ}}(\tau) + (1-\theta) p_{\text{succ}}(\sigma_0)
	\end{align} and $P(1 \mid \theta) = 1- P(0 \mid \theta)$.
	
	Since $P(i \mid \theta)$, $i\in\{0,1\}$ is obtained from a series of quantum maps followed by a projection on $\tau$, we see that this fits into the basic form $P(i \mid \theta) = \Tr[\Phi_\theta(\tau)M_i]$. Suppose we would like to estimate the parameter $\theta$. We can then use Eq.~\ref{def::classFish} to evaluate the classical information of the statistics.
	
	One may verify that 
	\begin{align}
	F_C(\tau \mid M) = \frac{\left(\pdv{P(0 \mid \theta)}{\theta}\right)^2}{P(0 \mid \theta)(1-P(0 \mid \theta))}. \label{eq::Thm1FC}
	\end{align}
	
	Evaluating near the vicinity of $\theta = 0$, the resulting FI is 
	\begin{align}
	F_C(\tau \mid M) = \frac{[p_{\text{succ}}(\tau)- p_{\text{succ}}(\sigma_0)]^2}{p_{\text{succ}}(\sigma_0)(1-p_{\text{succ}}(\sigma_0))}.
	\end{align}
	
	Recall that so far, the state $\sigma_0$ is not yet specified. We now choose it such that it satisfies $p_{\text{succ}}(\sigma_0) = \min_{\sigma \in  \free} p_{\text{succ}}(\sigma)$. This means that the numerator is a monotonically increasing function of $p_{\text{succ}}(\tau)$. We also see that the denominator does not depend on the state $\tau$. Together with the fact that  $p_{\text{succ}}(\rho) > p_{\text{succ}}(\sigma)$ for every $\sigma \in  \free$, we must have $F_C(\rho \mid  M) > \max_{\sigma \in  \free} F_C(\sigma \mid M)$
	and $N_C(\rho \mid M) > 0 $ if $\rho$ is nonclassical. This shows the existence of at least one parameter estimation problem where $N_C(\rho \mid M) > 0 $ if $\rho$ is nonclassical. For the special case where $p_{\text{succ}}(\sigma_0) = 0$, then Eq.~\ref{eq::Thm1FC} becomes $F_C(\tau \mid M) = p_{\text{succ}}(\tau)/\theta$ instead and a similar conclusion is reached. This is sufficient to prove both directions of the statement for $N_C$.
	
	The equivalent statement for $N_Q$ comes from the observation that the  quantum map $\Lambda_4 \circ \Lambda_3 \circ \Lambda_2 \circ \Lambda_1$ maps any input state to a diagonal state, and that the measurement $M$ is also diagonal. We then use the fact that for diagonal states, the QFI is saturated by a measurement in the diagonal basis~\cite{Tan2019}. This is sufficient to show that $F_C(\rho \mid  M) - \max_{\sigma \in  \free} F_C(\sigma \mid M) = F_Q(\rho ) - \max_{\sigma \in  \free} F_Q(\sigma ) = N_Q(\rho >0)$, which proves the required statement.
	\end{proof}

Theorem~\ref{thm::existence} establishes that $N_C$ and $N_Q$ are both able to identify any resourceful state in general quantum resource theories. Both of the quantities therefore constitute faithful resource witnesses of direct physical relevance. This result also demonstrates the existence of a metrological advantage for any nonclassical state.

However, one can also be interested in understanding this advantage precisely: \textit{how much} advantage can be extracted from a given state $\rho$ over all resourceless states?
To this end, an even stronger statement can be proven which quantitatively relates the quantum advantage $N_C$ and an important resource quantifier --- the (generalized) robustness measure~\cite{vidal_1999}, which we denote $R(\rho)$.

So far, we have considered $N_C$ given some fixed parameter estimation problem with encoding $\Phi_\theta$ and measurement $M$.
As a measure of the extent of the quantum advantage, it is reasonable to consider the maximum nonclassical advantage one may obtain over all such encodings and measurements. Since the Fisher information can be scaled via a simple reparametrization $\theta \rightarrow k\theta$, we are also motivated to normalize the type of encoding channels over the set of free states. In light of these considerations, we can define the following  quantity:
\begin{align}
N^{\max}_C(\rho) \coloneqq \max_{\Phi_\theta \in \mathcal{P}, M}  N_C(\rho \mid \Phi_\theta , M ),
\end{align} where  $\mathcal{P}$ is the set of all parameter estimation problems satisfying $\max_{\sigma \in  \free} F_C(\sigma \mid \Phi_\theta, M ) \leq 1$. It can be shown that $N^{\max}_C$ is a resource monotone, which we elaborate further upon in the Supplemental Material~\cite{CompanionPaper}.

\begin{theorem} \label{thm::NCbounds}
	In any resource theory, there exists a parameter estimation task $\Phi_\theta$ such that 
	\begin{equation}\begin{aligned}
	\R(\rho)^2 \leq N_C(\rho \mid \Phi_\theta, M) \leq \R(\rho)^2 + 2 \R(\rho),
	\end{aligned}\end{equation}
	
	In particular, $N_C(\rho \mid \Phi_\theta , M) > 0$ iff $\rho \notin \free$, and $N_C^{\max}(\rho) \geq \R(\rho)^2$ is a computable lower bound for any resource.
\end{theorem}

Theorem~\ref{thm::NCbounds} provides a computable lower bound on the quantum advantage that can be extracted. 
We stress that the robustness $R(\rho)$ can always be computed as a convex optimization problem. In many cases, such as the resource theories of coherence~\cite{piani_2016}, multi-level coherence~\cite{ringbauer_2018}, and magic~\cite{veitch_2012,howard_2017,wang_2020,seddon_2020}, it becomes an efficiently computable semidefinite program, while in many theories including entanglement~\cite{vidal_1999,steiner_2003,harrow_2003} and multi-level entanglement~\cite{johnston_2018} it can be computed analytically for all pure states. Furthermore, in the class of affine resource theories~\cite{gour_2017,regula_2019}, which includes theories such as coherence and imaginarity~\cite{hickey_2018,wu_2021}, the lower bound of Thm.~\ref{thm::NCbounds} is tight, in the sense that there always exists a task such that $N_C(\rho \mid \Phi_\theta, M) = \R(\rho)^2$.

Taking this quantitative relationship further, it is natural to ask whether there is also an upper bound on the quantum advantage that a resource state affords in estimation tasks. Indeed, this is possible in the case where the decoding measurement has a binary outcome:
\begin{theorem}\label{thm::ubound}
    For any parameter estimation task with encoding channel family $\Phi_\theta$ and two\hyp outcome measurement $M\equiv(P,\openone-P)$, let $r\coloneqq F_C(\rho \mid  M)<\infty$ and
    \begin{equation}
        \omega:=\left|\pdv{\Tr{P\Phi_\theta(\rho)}}{\theta}\right|_{\theta=0}.
    \end{equation}
    Then,
\begin{align}\label{ubound}
    N_C(\rho \mid M)&\le r-\frac{\left[R_S(\rho)+1\right]^{-2}\omega^2}{\max\limits_{\tau\in\free}\Tr{\left[P\,\Phi_0(\tau)\right]}\left(1-\Tr{\left[P\,\Phi_0(\tau)\right]}\right)}\nonumber\\
    &\le r-\frac{4\omega^2}{\left[R_S(\rho)+1\right]^2},
\end{align}
where $R_S(\rho)$ is the standard robustness of $\rho$ with respect to $\free$~\cite{vidal_1999}.
\end{theorem}
$R_S$ is another operationally significant resource measure closely related to $R$, and like the latter, admits efficient SDP formulations in many resource theories. Thus, given any particular parameter estimation task with a two\hyp outcome measurement, we have an efficiently computable upper bound to the quantum advantage $N_C$ of any given resource state. The bound in the first line is less computationally feasible but tighter.

We note here that the parameter $\omega$ scales with the energy cost of applying the given family of encoding channels on $\rho$, with finer $\theta$\hyp resolution costing more energy. We expand on this connection in~\cite{CompanionPaper}, together with a proof of Theorem~\ref{thm::ubound}; there we also discuss why we suspect the upper bound \eqref{ubound} cannot be made independent of the estimation task in general. Finding upper bounds for the case of more general measurements is left for future work.
\vspace{2ex}

\noindent{\it Quantum resources for unitary encodings ---} We now consider the important special case where the quantum channel $\Phi_\theta$  is a unitary encoding channel $\Phi_\theta(\rho) = U_\theta \rho U_\theta^\dagger$ and $U_\theta \coloneqq e^{-i\theta G}$. Here, $G$ is some Hermitian operator specifying the unitary evolution, and is called the generator of the unitary encoding.

Given any unitary encoding generated by the Hermitian operator $G$, one may be interested to know whether the parameter estimation problem corresponding to $G$ is benefited by having nonclassical states in a given quantum resource theory. The following result establishes a simple criterion for determining whether $G$ reveals nonclassicality.

\begin{theorem} \label{thm::GCriterion}
	Consider any Hermitian generator $G$ and convex set of free states $\free$. Let $s^*$ be an optimal solution to the convex optimization problem 	
	

\begin{equation*}
\begin{aligned}
& \underset{X \geq 0}{\text{\rm maximize}}
& & 2\Tr[(G_A \otimes \openone_B - \openone_A \otimes G_B )^2 X_{AB}] \\
& \text{\rm subject to}
& & \Tr_A X_{AB} = \Tr_B X_{AB} = \sigma \in \free.
\end{aligned}
\end{equation*}

Let $\lambda_{\max}, \lambda_{\min}$ denote the largest and smallest eigenvalues of the generator $G$ resepectively. If $(\lambda_{\max} - \lambda_{\min})^2 > s^*$, then $N_Q(\rho) > 0$ for some $\rho$.
\end{theorem}

\begin{proof}
One may immediately verify that the objective function $2\Tr[(G_A \otimes \openone_B - \openone_A \otimes G_B )^2 X_{AB}]$ is linear and that the feasible set is convex, so $s^*$ is the solution to a convex optimization problem.

To prove the statement, we just need to show that $s^*$ upper bounds $\max_{\sigma \in  \free} F_Q(\sigma)$. In general, for any generator $G$ the maximum achievable QFI can be verified to be $(\lambda_{\max}-\lambda_{\min})^2$ \cite{Giovannetti2006}, so if $(\lambda_{\max}-\lambda_{\min})^2 > s^* \geq \max_{\sigma \in  \free} F_Q(\sigma)$, we necessarily have $N_Q(\rho) > 0$.

To see that $s^*$ in indeed an upper bound, we use the fact that $F_Q(\sigma) \leq 4 \Delta^2_\sigma G$ where $\Delta^2_\sigma G$ is the variance of $G$ given the state $\sigma$. We then observe that $X_{AB} = \sigma \otimes \sigma$ where $\sigma \in \free$ is a feasible solution. Finally, we observe that $2\Tr[(G_A \otimes \openone_B - \openone_A \otimes G_B )^2 \sigma \otimes \sigma] = 4  \Delta^2_\sigma G$ so we must have $s^* \geq \max_{\sigma \in  \free} F_Q(\sigma)$ as required.
\end{proof}

The convex optimization in Theorem~\ref{thm::GCriterion} provides a direct method of testing of whether the particular parameter estimation problem generated by $G$ will benefit from a given quantum resource. Note that this is a sufficient condition, so failure of the test does not necessarily imply the resource is not useful for this encoding. We also highlight that the criterion is based on the QFI\hyp based quantity, $N_Q$, but it also applies to the CFI case as well, since if $N_Q(\rho) > 0$, then we are guaranteed some measurement $M$ for which $N_C(\rho \mid M) > N_Q(\rho) > 0$. A similar sufficient condition can be obtained for non\hyp unitary encoding channels, in terms of their unitary dilation.

We illustrate Theorem~\ref{thm::GCriterion} with a simple worked example. Consider a qubit system, with the free set $\free =  \{ \ketbra{0} \}$ being a trivial set with only a single element. Let $G = \sigma_z$, the Pauli matrix in the $z$ direction. In this case, the feasible set only has one state $X_{AB} = \ketbra{0}\otimes \ketbra{0}$, from which we can verify $s^* = 0$. Since $(\lambda_{\max}- \lambda_{\min})^2 = 4$, from Theorem~\ref{thm::GCriterion}, there must exist some state $\rho \not{\in} \free$ such that $N_C(\rho) > 0$. One can repeat the same argument for the case where $\free =  \{ \ketbra{1} \}$. Since the FI is convex~\cite{Sidhu2020}, $\rho$ must outperform any convex mixture of $\ketbra{0}$ and $\ketbra{1}$, i.e.\ any incoherent quantum state. This is one way to verify that quantum coherence is a useful nonclassical resource for the unitary encoding generated by $G = \sigma_z$.\vspace{2ex}

\noindent{\it Identifying nonclassical operations ---} Thus far, we have considered the use of the FI to identify general resources in quantum states. Recall that every quantum resource theory is also accompanied by some set of free quantum operations $\mathcal{O}$. Just as any resourceful state can be defined as a state that is not in the set of free states $\free$, we can similarly define a resourceful operation as any quantum channel not in the set $\mathcal{O}$~\cite{liu_2019-1,liu_2020}. It turns out that the FI can also universally distinguish operations with and without a given resource.

\begin{theorem} \label{thm::nonclassOp}
	For any set of free operations $\mathcal{O}$ and quantum map $\Xi \not \in  \mathcal{O}$, there exists a quantum trajectory $\rho_\theta$ on an extended Hilbert space such that the map $\openone \otimes \Xi$ satisfies
	\begin{align}
	F_C[\openone \otimes \Xi (\rho_\theta ) \mid M] > \max_{\Omega \in  \free} F_C[\openone \otimes \Omega (\rho_\theta ) \mid M]  
	\end{align} for some $\theta$.
\end{theorem}
A full discussion of the proof can be found in~\cite{CompanionPaper}. Theorem~\ref{thm::nonclassOp} demonstrates that the FI plays a foundational role not just in the study of nonclassicality in states, but also in the study of the resources of quantum channels.\vspace{2ex}

\noindent{\it Conclusion ---} 
Many quantities traditionally used to study the nonclassical features of quantum mechanics are typically relevant only when considering specific notions of nonclassicality. However, quantum advantages in different tasks rely on a broad range of quantum phenomena, motivating the study of physical quantities which can be used to identify \textit{all} reasonable quantum resources.
To this end, we showed that two Fisher information (FI)\hyp based quantities $N_C$ and $N_Q$, defined through the classical and quantum variants of the FI, respectively, are examples of such universally relevant operational quantities. This implies that every quantum resource can provide a quantum advantage in some parameter estimation problem. In this sense, any feasible notion of nonclassicality objectively always provides a quantum advantage in metrology, although it remains subjective as to whether such applications are relevant to the interests of an experimentalist. We also highlight that, while the focus of this work is on identifying (detecting) general quantum resources, it is also possible to construct resource measures --- also called resource monotones --- using $N_C$ and $N_Q$. This is discussed in greater detail in the Supplemental Material~\cite{CompanionPaper}.

We then provided a lower bound on the maximum extent of this quantum advantage in terms of the generalized robustness measure~\cite{ vidal_1999, takagi_2019, takagi_2019-2}, which also universally identifies resourceful states. For the case of estimation problems with binary outcomes, we also provided an upper bound, in terms of a related quantity called the standard robustness. The special case of unitary encodings was also considered, where we provided a simple criterion to test whether a given quantum resource provides an advantage for a unitary encoding generated by a Hermitian operator $G$. Finally, we showed that not only does the FI universally identify resources of quantum states, it also universally witnesses resourceful quantum operations in every resource theory. These results solidify the central role that the FI plays in the study of quantum resources.\vspace{2ex}

\noindent\textit{Acknowledgments ---} K.C.T.\ and B.R.\ acknowledge support by the NTU Presidential Postdoctoral Fellowship program funded by Nanyang Technological University. V.N.\ acknowledges support from the Lee Kuan Yew Endowment Fund (Postdoctoral Fellowship).

\bibliography{fisher_prl}

\onecolumngrid
\vspace{2\baselineskip}
\twocolumngrid

\section*{Supplemental Material}

In this Supplemental Material, we provide detailed technical proofs of several theorems in the main text, as well as additional discussions comparing classical and quantum Fisher information, as well as how the Fisher information may be used to construct nonclassicality measures in general resource theories.

\subsection{Proof of Theorem~\ref{thm::NCbounds}}

We employ the (generalized) robustness measure $R$~\cite{vidal_1999}, given by
\begin{equation}\begin{aligned}
R(\rho) = \inf \left\{ \lambda \;\left|\; \frac{\rho + \lambda \omega}{1+\lambda} \in \free \right.\right\},
\end{aligned}\end{equation}
where the optimization is over all quantum states $\omega$.

The result below is a slight extension of the statement of Thm.~\ref{thm::NCbounds} from the main text.

\begin{theorem*} 
	In any resource theory, there exists a parameter estimation task $\Phi_\theta$ such that 
	\begin{equation}\begin{aligned}
	\R(\rho)^2 \leq N_C(\rho \mid \Phi_\theta, M) \leq \R(\rho)^2 + 2 \R(\rho),
	\end{aligned}\end{equation}
	In particular, $N_C(\rho \mid \Phi_\theta , M) > 0$ iff $\rho \notin \free$, and $N_C^{\max}(\rho) \geq \R(\rho)^2$ is a computable lower bound for any resource.
	
		Furthermore, if the resource theory is affine --- that is, the set $\free$ is the intersection of some affine subspace with the set of density matrices --- then $ N_C(\rho \mid \Phi_\theta, M) =\R(\rho)^2 $.
\end{theorem*}

\begin{proof}
We will use the fact that the robustness can be written as~\cite{brandao_2005,regula_2018}
\begin{equation}\begin{aligned}\label{eq:rob_dual}
  \R (\rho) = \sup \left\{\left. \Tr(W\rho) - 1 \;\right|\; W \geq 0,\; \Tr(W \sigma) \leq 1 \; \forall \sigma \in \free \right\},
\end{aligned}\end{equation}
and let $W \neq 0$ be any feasible solution to the above optimisation. As shown in~\cite[Thm.\ 2]{takagi_2019-2}, there exists a channel discrimination task with $d$ channels $\{\Lambda_i\}$ (in fact, unitary channels) and a POVM $N = \{N_i\}$ such that the probability of successfully discriminating the channels with uniform prior probability for any input state $\tau$ is
\begin{equation}\begin{aligned}
  p_{\succ}(\tau) \coloneqq \sum_i \frac{1}{d} \Tr(\Lambda_i(\tau) N_i) = \frac{\Tr(W\tau)}{\Tr W}.
\end{aligned}\end{equation}
We then define $\sigma_0 \coloneqq \operatorname{arg min}_{\sigma \in \free} \Tr(W\sigma)$.

We will first assume that $p_\succ(\sigma_0) >0$, so that $\Tr(W\sigma_0) > 0$. Noticing that $\Tr W \geq \Tr(W \sigma_0)$, we then define the channel $\Phi_\theta$ as
\begin{widetext}
\begin{equation}\begin{aligned}
  \mathcal{A}_\theta (\tau) \coloneqq& \left(\sum_i \frac{1}{d} \Lambda_i (\tau) \otimes \proj{i} \right) \otimes \sqrt{\Tr (W \sigma_0) (\Tr W - \Tr (W \sigma_0))}\, \theta \proj{0}\\
  +&  \left(\sum_i \frac{1}{d} \Lambda_i (\sigma_0) \otimes \proj{i} \right) \otimes (1- \sqrt{\Tr (W \sigma_0) (\Tr W - \Tr (W \sigma_0))} \, \theta) \proj{1},\\
    \Phi_\theta (\tau) \coloneqq&   \Tr\left( \mathcal{A}_\theta(\tau) \left[ \sum_i  N_i \otimes \proj{i} \otimes \openone \right] \right) \proj{0} +  \Tr\left( \mathcal{A}_\theta(\tau) \left[ \sum_{i \neq j}  N_i \otimes \proj{j} \otimes  \openone \right] \right) \proj{1}
\end{aligned}\end{equation}
\end{widetext}
where, without loss of generality, we have taken $\theta$ small enough so that $\mathcal{A}_\theta$ and $\Phi_\theta$ are valid CPTP maps. Performing the measurement $M = \{ \proj{0}, \proj{1} \}$, we then obtain the statistics
\begin{equation}\begin{aligned}
  P(0 \mid \theta) &= \theta \sqrt{\Tr (W \sigma_0) (\Tr W - \Tr (W \sigma_0))}\, p_{\succ}(\tau) \\
  & + \left(1-\theta \sqrt{\Tr (W \sigma_0) (\Tr W - \Tr (W \sigma_0))}\right) p_{\succ}(\sigma_0),\\
  P(1 \mid \theta) &=  1 - P(0|\theta).
\end{aligned}\end{equation}
We can then use Eq.~\eqref{def::classFish} to compute the classical information of the statistics. One may verify that 
	\begin{align}
	F_C(\tau \mid \Phi_\theta, M) = \frac{\left(\pdv{}{\theta} P(0 \mid \theta)\right)^2}{P(0 \mid \theta)(1-P(0 \mid \theta))}.
	\end{align}
Evaluating near the vicinity of $\theta = 0$, the resulting Fisher information is then
\begin{equation}\begin{aligned}
  &F_C(\tau \mid \Phi_\theta, M) \\
  &= \frac{\left[\Tr (W \sigma_0) (\Tr W - \Tr (W \sigma_0))\right][p_{\succ}(\tau) - p_{\succ}(\sigma_0)]^2}{p_{\succ}(\sigma_0)(1-p_{\succ}(\sigma_0))}\\
  &= \frac{\left[\Tr (W \sigma_0) (\Tr W - \Tr (W \sigma_0))\right]\left[\frac{\Tr(W\tau) - \Tr(W\sigma_0)}{\Tr W}\right]^2}{\frac{\Tr(W\sigma_0)}{\Tr W}\left(1 - \frac{\Tr(W\sigma_0)}{\Tr W}\right)}\\
  &= [\Tr(W \tau) - \Tr(W\sigma_0)]^2.
\end{aligned}\end{equation}

In the case that $p_\succ(\sigma_0)= 0$, we can make an analogous choice of channels $\mathcal{A}_\theta, \Phi_\theta$ but replacing the constant $\sqrt{\Tr (W \sigma_0) (\Tr W - \Tr (W \sigma_0))}$ with $\Tr W$. This then gives the statistics
\begin{equation}\begin{aligned}
  P(0 \mid \theta) &= \theta \Tr W\, p_{\succ}(\tau)\\
  P(1 \mid \theta) &=  1 - \theta \Tr W \, p_{\succ}(\tau).
\end{aligned}\end{equation}
Taking $\theta$ to 0, this gives
\begin{equation}\begin{aligned}
	F_C(\tau \mid \Phi_\theta, M) &= [ \Tr W p_{\succ}(\tau)]^2\\
	&= [\Tr(W \tau)]^2\\
	&= [\Tr(W \tau) - \Tr(W\sigma_0)]^2
\end{aligned}\end{equation}
analogously to the case $p_\succ(\sigma_0) >0$.

Noticing that $F_C(\sigma \mid \Phi_\theta, M) \leq 1$ for any $\sigma \in \free$, we have that $\Phi_\theta \in \mathcal{P}$, and so
\begin{equation}\begin{aligned}
	N_C^{\max}(\rho) &\geq N_C(\rho | \Phi_\theta, M) \\
	&= [\Tr(W \rho) - \Tr(W\sigma_0)]^2 - \max_{\sigma \in \free} [\Tr(W \sigma) - \Tr(W\sigma_0)]^2.
\end{aligned}\end{equation}
In the degenerate case $\R(\rho) = \infty$, $\Tr(W\rho)$ can be made arbitrarily large, so we see that $N_C^{\max}(\rho) = \infty$. Assuming then that $\R(\rho) < \infty$, there must exist an optimal choice $W^\star$ such that $\Tr(W^\star \rho) - 1 = \R(\rho)$; for such a choice, there must then exist $\sigma \in \free$ such that $\Tr(W^\star\sigma) = 1$, since otherwise we could rescale $W^\star$ to obtain a higher value of the robustness, contradicting its optimality. With this choice of $W$, we obtain
\begin{equation}\begin{aligned}
  F_C(\rho \mid \Phi_\theta, M) &= [\R(\rho) + 1 - \Tr(W^\star\sigma_0)]^2,\\
  \max_{\sigma \in \free} F_C(\sigma \,|\, \Phi_\theta, M) &= [1 - \Tr(W^\star\sigma_0) ]^2,
\end{aligned}\end{equation}
where we note that the terms inside the outer brackets are always non-negative. This gives
\begin{equation}\begin{aligned}
	&N_C(\rho \mid \Phi_\theta, M)\\
	&= [\R(\rho) + 1 - \Tr(W^\star\sigma_0)]^2 - [1 - \Tr(W^\star\sigma_0) ]^2\\
	&= \left[\R(\rho) + 1 - \Tr(W^\star\sigma_0) - 1 + \Tr(W^\star\sigma_0)\right]\\
	&\quad\times \left[\R(\rho) + 1 - \Tr(W^\star\sigma_0) + 1 - \Tr(W^\star\sigma_0)\right]\\
	&= \R(\rho) \left[ \R(\rho) + 2(1-\Tr(W^\star\sigma_0)) \right].
\end{aligned}\end{equation}
Since $\Tr(W^\star \sigma_0) \in [0,1]$, we thus see that
\begin{equation}\begin{aligned}
	\R(\rho)^2 \leq N_C(\rho \mid \Phi_\theta, M) \leq \R(\rho)[ \R(\rho) + 2] .
\end{aligned}\end{equation}

In the special case of affine resource theories, one can always choose an optimal $W^\star$ satisfying $\Tr(W^\star\sigma) = 1 \; \forall \sigma \in \free$~\cite{regula_2019}, which gives $N_C(\rho | \Phi_\theta, M) = R(\rho)^2$.

\end{proof}

\begin{remark}
We have seen in the proof that the lower bound $N_C(\rho \mid \Phi_\theta, M) \geq \R(\rho)^2$ can be tight. Note that there also exist cases in which the upper bound is tight, in the sense that $N_C(\rho \mid \Phi_\theta, M) = \R(\rho)^2 + 2 \R(\rho)$. This includes the resource theory of entanglement where, for a $d \times d$-dimensional pure state $\rho$, the optimal witness $W^\star$ is simply $W^\star = d \proj{\Psi^+}$ with $\ket{\Psi^+} = \frac{1}{\sqrt{d}} \sum_{i} \ket{ii}$~\cite{harrow_2003,steiner_2003}, which satisfies $\Tr(W^\star \sigma_0) = 0$.

We also remark that the degenerate case $\R(\rho) = \infty$, which we have considered for completeness, can only occur when the support of $\rho$ is not contained in the support of any state in $\free$. In particular, as long as there exists at least one full-rank state $\sigma \in \free$, we always have $\R(\rho) < \infty$.
\end{remark}

\subsection{Proof of Theorem~\ref{thm::ubound}}

Here, we employ the standard robustness~\cite{vidal_1999}
\begin{equation}\begin{aligned}
R_S(\rho) = \inf \left\{ \lambda \;\left|\; \frac{\rho + \lambda \sigma}{1+\lambda} \in \free,\; \sigma \in \free \right.\right\}.
\end{aligned}\end{equation}

\begin{theorem*}
For any parameter estimation task with encoding channel family $\Phi_\theta$ and two\hyp outcome measurement $M\equiv(P,\openone-P)$, let $r\coloneqq F_C(\rho \mid  M)<\infty$ and
    \begin{equation}
        \omega:=\left|\pdv{\Tr{\left[P\,\Phi_\theta(\rho)\right]}}{\theta}\right|_{\theta=0}.
    \end{equation}
    Then,
\begin{align}\label{ubounda}
    N_C(\rho \mid M)&\le r-\frac{\left[R_S(\rho)+1\right]^{-2}\omega^2}{\max\limits_{\tau\in\free}\Tr{\left[P\,\Phi_0(\tau)\right]}\left(1-\Tr{\left[P\,\Phi_0(\tau)\right]}\right)}\nonumber\\
    &\le r-\frac{\left[R_S(\rho)+1\right]^{-2}\omega^2}{\max\limits_{\tau\in\free}\Tr{\left[P\,\Phi_0(\tau)\right]}\left(1-\min\limits_{\eta\in\free}\Tr{\left[P\,\Phi_0(\eta)\right]}\right)},
\end{align}
where $R_S(\rho)$ is the standard robustness of $\rho$ with respect to $\free$.
\end{theorem*}
\begin{proof}
For a given arbitrary two\hyp outcome estimation task $\left[\Phi_\theta, M\equiv(P, \openone-P)\right]$ and probe state $\tau$, define
\begin{equation}\label{defptau}
p_\tau\left(\theta\right):=\Tr\left[P\,\Phi_\theta(\tau)\right].
\end{equation}
as a function of the real parameter $\theta$. We then observe that
\begin{equation}\label{fcpt}
F_C(\tau \mid \Phi_\theta , M )=\frac{\dot p^2_\tau\left(0\right)}{p_\tau(0)\left[1-p_\tau(0)\right]},
\end{equation}
where a dot above denotes the derivative wrt $\theta$. For brevity, we will omit the argument $0$ for $p_\tau$ and $\dot p_\tau$ in the remainder, explicating only nonzero arguments where they occur. We first note that
\begin{align}\label{numden}
\max_{\sigma \in  \free}F_C(\sigma \mid \Phi_\theta , M )&=\max_{\sigma \in  \free}\frac{\dot p^2_\sigma}{p_\sigma\left(1-p_\sigma\right)}\nonumber\\
&\ge\frac{\max\limits_{\sigma \in  \free}\dot p^2_\sigma}{\max\limits_{\tau \in  \free}p_\tau\left(1-p_\tau\right)}
\end{align}

We will now lower\hyp bound the numerator. To this end, consider the task of discriminating between the two channels (occurring with equal probabilities) $\Lambda_0:=M\circ\Phi_0$ and $\Lambda_1:=M\circ\Phi_\delta$ for some small $\delta>0$, where $M$ is co\hyp opted to denote the channel $M(X):=\Tr\left(PX\right)\proj0+\left[1-\Tr\left(PX\right)\right]\proj1$. Using a probe state $\tau$, the game reduces to that of distinguishing the states
$$p_\tau\left(k\delta\right)\proj0+\left[1-p_\tau\left(k\delta\right)\right]\proj1,$$
$k\in\{0,1\}$. By the Haelstrom bound, the ``bias'' or ``gain'' that the best distinguishing strategy achieves over a random guess (what \cite{takagi_2019} denotes by variants of the notation $p_{\mathrm{gain}}$) is proportional to $\delta\left|\dot p_\tau\right|$
for small $\delta$. Now we invoke \cite[Thm.\ 7]{takagi_2019}, which puts an upper bound on the gain of a given resource state relative to that achieved by the free states:
\begin{equation}\label{gbound}
\frac{p_{\mathrm{gain}}(\rho)}{\max\limits_{\sigma\in\free}p_{\mathrm{gain}}(\sigma)}\le R_S(\rho)+1.
\end{equation}
By the above observation relating $p_{\mathrm{gain}}$ to $\dot p_\tau$, it follows that
\begin{equation}\label{eqnbound}
\max_{\sigma \in  \free}\dot p^2_\sigma\ge\frac{\dot p^2_\rho}{\left[R_S(\rho)+1\right]^2}=\frac{\omega^2}{\left[R_S(\rho)+1\right]^2}.
\end{equation} The result of the theorem in the main text then follows by noting that $\max\limits_{\tau \in  \free}p_\tau\left(1-p_\tau\right)\le0.25$, together with Eq.~\eqref{numden}.
\end{proof}

\begin{remark}[$\omega$ as an energy scale]
To understand the significance of the parameter $\omega$, let us consider an encoding channel family given explicitly through its Stinespring dilation, as
\begin{equation}
    \Phi_\theta\left(\rho\right)=\Tr_C\left[e^{-i\theta H_{AB}}\left(\rho_A\otimes\proj\psi_B\right)e^{i\theta H_{AB}}\right],
\end{equation}
where $H_{AB}$ is Hermitian, $C$ is some subsystem of $AB$, and $\ket\psi$ can be made $\theta$\hyp independent as above without loss of generality. For a small $\theta$, we have
\begin{equation}
    \Phi_\theta\left(\rho\right)\approx\tilde\rho+i\theta G_\rho,
\end{equation}
where $\tilde\rho\equiv\Tr_C\left(\rho_A\otimes\proj\psi_B\right)$ is $\theta$\hyp independent, and
\begin{equation}
    G_\rho=\Tr_C\left[\rho_A\otimes\proj\psi_B,H_{AB}\right]
\end{equation}
generates the action of $\Phi_\theta$ in the neighbourhood of $\theta=0$. We thus have
    \begin{equation}
        \omega=\left|\pdv{\Tr{\left[P\,\Phi_\theta(\rho)\right]}}{\theta}\right|_{\theta=0}=\left|\Tr{\left(iP\,G_\rho\right)}\right|.
    \end{equation}
Since $H$ is the Hamiltonian generating the encoding operation on $\rho$, we see that $\omega$ is a measure of the order of magnitude of the energy used in executing the encoding on the support of $\rho$ measured by $P$. In particular, the same encoding channel family can be implemented at a finer resolution by using a Hamiltonian $H'=cH$ (corresponding to $\omega'=c\omega$) for $c>1$.
\end{remark}

\begin{remark}[more general bounds]
Can we make the bound even more task\hyp independent than in Theorem~\ref{thm::ubound}? Unfortunately, this is unlikely: while convex\hyp geometric properties such as robustness measures help upper\hyp bound linear functionals like $p_\tau$ for a resource state in terms of the maximum value attained in the convex free set $\free$, the nonlinear $p_\tau(1-p_\tau)$ appears to have no such task\hyp independent bounds. A possible recourse might be to use convex\hyp geometric bounds on each of the two factors in the denominator in the last expression in \eqref{ubounda}; however, the resulting bound turns out to be worse than the $0.25$ of the main theorem. On the other hand, by eschewing generality and considering tasks where the performance of $\rho$ necessarily ``costs'' at least $\omega$, we prevent $p_\rho(1-p_\rho)$ from getting vanishingly small and helping $\rho$ attain arbitrarily high FI at a fixed energy cost.
\end{remark}

\subsection{Proof of Theorem~\ref{thm::nonclassOp}}

\begin{theorem*}
	For any set of free operations $\mathcal{O}$ and quantum map $\Xi \not \in  \mathcal{O}$, there exists a quantum trajectory $\rho_\theta$ on an extended Hilbert space such that the map $\openone \otimes \Xi$ satisfies
	\begin{align}
	F_C[\openone \otimes \Xi (\rho_\theta ) \mid M] > \max_{\Omega \in  \free} F_C[\openone \otimes \Omega (\rho_\theta ) \mid M]  
	\end{align} for some $\theta$.
\end{theorem*}

\begin{proof}
	We will use the fact that for any set of free operations $\mathcal{O}$ and quantum channel $\Xi \not \in \mathcal{O}$, there exists some collection of states and probabilities $\{p_i, \rho_i \}$, as well as POVMs $\{\pi_j\}$ such that $p_{\text{succ}}(\Xi) > \max_{\Omega \in \mathcal{O}} p_{\text{succ}}( \Omega)$, where $p_{\text{succ}}(\Xi) \coloneqq \sum_i p_i\Tr[\openone \otimes \Xi (\rho_i) \pi_i]$(Ref.~\cite{takagi_2019-2}).
	
	Consider the following input state:
	\begin{align}
	\rho_\theta   \coloneqq \sum_i \rho_i \otimes p_i \ketbra{i} \otimes [\theta \ketbra{0}+(1-\theta) \ketbra{1}].
	\end{align} For some $\Omega_0 \in \mathcal{O}$, we consider the following series of quantum maps that depends on the map $\Xi$:
	
	\begin{align}
	&\Lambda_1  (\rho_\theta )    \\
	&\coloneqq \theta \sum_i \openone \otimes \Xi(\rho_i) \otimes p_i \ketbra{i} \otimes \ketbra{0} \\
	&\quad + (1-\theta) \sum_i \openone \otimes \Omega_0 (\rho_i)  \otimes p_i \ketbra{i} \otimes  \ketbra{1}.
	\end{align} Here, the map $\Lambda_1$ performs the map $\openone \otimes \Xi$ on $\rho_i$ when the last qubit is in the state $\ket{0}$ and performs some other operations $\openone \otimes \Omega_0$ where $\Omega_0 \in \mathcal{O}$ when the last qubit is in the state $\ket{1}$. We then consider the quantum map
	
	\begin{align}
	&\Lambda_2 \circ \Lambda_1  (\rho_\theta )  \\
	&  \coloneqq \Tr \{ \Lambda_1 (\rho_\theta )  \pi_i \otimes \ketbra{i} \otimes \openone \} \ketbra{0} \\
	&\quad + \Tr \{ \Lambda_1 (\rho_\theta )   \pi_i \otimes (\openone -  \ketbra{i}) \otimes \openone \} \ketbra{1}.
	\end{align}
	
We then perform the projection $P_0 \coloneqq \ketbra{0}$ and $P_1 \coloneqq \ketbra{1}$ to obtain the statistics 
\begin{align}
P(0 \mid \theta) = \theta p_{\text{succ}}(\Xi) + (1-\theta) p_{\text{succ}}(\Omega_0)
\end{align} and 
\begin{align}
P(1 \mid \theta) = 1- P(0 \mid \theta).
\end{align}

We can absorb the maps $\Lambda_1$, $\Lambda_2$ and measurements $P_0,P_1$ and represent it as a single generalized measurement $M = \{ M_i \}$. One may then verify that 
\begin{align}
F_C[\openone \otimes \Xi (\rho_\theta ) \mid M] = \frac{(\pdv{P(0 \mid \theta)}{\theta})^2}{P(0 \mid \theta)(1-P(0 \mid \theta))}.
\end{align}

Similar to Theorem~\ref{thm::existence}, evaluating near the vicinity of $\theta = 0$, the resulting Fisher information is 
\begin{align}
F_C[\openone \otimes \Xi (\rho_\theta ) \mid M] = \frac{[p_{\text{succ}}(\Xi)- p_{\text{succ}}(\Omega_0)]^2}{p_{\text{succ}}(\Omega_0)(1-p_{\text{succ}}(\Omega_0))}.
\end{align}

The quantum map $\Omega_0$ is not yet specified. We choose it such that it satisfies $p_{\text{succ}}(\Omega_0) = \min_{\Omega \in  \mathcal{O}} p_{\text{succ}}(\Omega)$. This means that $F_C[\openone \otimes \Xi (\rho_\theta ) \mid M]$ is a monotonically increasing function of $p_{\text{succ}}(\Xi)$. Together with the fact that  $p_{\text{succ}}(\Xi) > p_{\text{succ}}(\Omega)$ for every $\Omega \in  \mathcal{O}$, we must have 
\begin{align}
F_C[\openone \otimes \Xi (\rho_\theta ) \mid M] > \max_{\Omega \in  \free} F_C[\openone \otimes \Omega (\rho_\theta ) \mid M]
\end{align} in the vicinity $\theta = 0$. This shows the existence of at least one quantum trajectory $\rho_\theta$ where $F_C[\openone \otimes \Xi (\rho_\theta ) \mid M] > \max_{\Omega \in  \free} F_C[\openone \otimes \Omega (\rho_\theta ) \mid M] $  for some value of $\theta$.
\end{proof}

\subsection{Comparison between classical and quantum Fisher information for resource identification}

In the main text, we make a distinction between classical and quantum Fisher information and defined corresponding  resource witnesses $N_C$ and $N_Q$. In was also discussed in the main text that the classical Fisher information is at least as good as the quantum Fisher information at identifying resource states. This is because for any parametrized channel $\Phi_\theta$, in the sense that it is always possible to find a measurement $M^\star$ such that $N_C (\rho \mid M^\star) \geq N_Q(\rho)$. In fact, there exist situations where $N_C$ is a strictly better resource identifier than $N_Q$. This is summarized by the following statement:
\begin{proposition}
There exist resource theories and unitary encodings of the form $U_\theta = e^{-i\theta G}$ where $N_C(\rho \mid M)$ identifies a strictly greater number of nonclassical states than $N_Q(\rho)$.
\end{proposition}

This can be shown by considering an explicit example. As mentioned above, recall that single parameter estimation problem, the set of nonclassical states identified by $N_C$ is at least as large as the set identified by $N_Q$.
It therefore suffices to demonstrate that there is at least one plausible set of classical states $ \free$ where $N_Q(\rho)=0$ for every unitary encoding $U_\theta = e^{-i\theta G}$, but there exists  Hermitian generator $G$ and measurement $M$ where $N_C(\rho \mid M) >0$.
	
To do this, let us consider an artificial but plausible resource theory with free states $ \free$. Consider a qubit system where $ \free$ is some strict convex subset of the Bloch sphere.

We define the subset $ \free \coloneqq \mathrm{Conv} \left ( \{ \ket{\phi, \theta} \mid \theta \in [\pi/2,3\pi/2] \} \right)$, where $\ket{\phi,\theta} \coloneqq \cos(\theta/2)\ket{0}+e^{i\phi}\sin(\theta/2)\ket{1}$. This is just the lower hemisphere of the Bloch sphere, inclusive of the $x$-$y$ plane. Since this is the convex hull over a set of pure states, it is also a strict convex subset of the Bloch sphere, which makes it a plausible set of classical states to consider for some resource theory based on the assumptions made in the main text.	
	
Consider now the set of all possible unitary encodings $U_\theta = e^{-i\mu G}$. Since the set of Pauli matrices $\{\openone, \sigma_x, \sigma_y, \sigma_z \}$ spans the operator space, we see that every Hermitian generator $G$ can be written as $r_0 \openone +\vec{r}.\vec{\sigma}$, where $\vec{\sigma} = \sigma_x, \sigma_y, \sigma_z)$ is the vector of Pauli matrices. Since $e^{-i\mu r_0 \openone}$ just adds a global phase, it does not affect any measurement statistics, so we can just consider a Hermitian generator of the type $G = \vec{r}.\vec{\sigma}$. Without any loss in generality, we can also normalize $G$ such that $\abs{\vec{r}} =1$. This means that $G = \vec{r}.\vec{\sigma}$ is a Pauli matrix pointing in the direction $\vec{r}$. To summarize, this means that every unitary encoding on the Bloch sphere is equivalent to a rotation about an axis $\vec{r}$.

For $G = \sigma_z$, it is known that every pure state lying on the $x$-$y$ plane, i.e.\ every state of the form $(\ket{0}+e^{i\phi} \ket{1})/\sqrt{2}$, maximizes the quantum Fisher information. We see that these states lie on a greater circle on the Bloch sphere. The problem of finding $G = \vec{r}.\vec{\sigma}$ such that $N_Q(\rho) > 0$ for some $\rho$ therefore boils down to finding a greater circle on the Bloch sphere which does not intersect the set of classical states $ \free$. If the greater sphere intersects $ \free$, then there exists a classical state which maximizes the quantum Fisher information, which necessarily means that $N_Q(\rho) \ngtr 0$. However, we see that it is impossible to draw a greater circle without intersecting the lower hemisphere. $N_Q$ is therefore unable to identify any nonclassical state outside of $ \free$ over the set of unitary encodings.

Let us consider $G = \sigma_y/2$ and POVMs $M_0 = \ketbra{0}+\ketbra{1}/2$ and $M_1 = \ketbra{1}/2$. First, we note that $G$ generates a rotation about the $y$-axis. Any component along the $y$ axis therefore does not contain any information about this rotation. It therefore suffices to optimize the classical Fisher information over the $z$-$x$ plane. One can then verify that the resulting classical Fisher information for any input state $\ket{\phi=0,\theta} \coloneqq \cos(\theta/2)\ket{0}+\sin(\theta/2)\ket{1}$ is given by  $F_C(\ketbra{\phi,\theta} \mid M) = 2\cos^2(\theta/2)/(3+\cos\theta)$. Over the top hemisphere, this is maximized by the state $\ket{0}$, and $F_C(\ketbra{0} \mid M) = 1/2$. Over the bottom hemisphere, this is maximized by $\ket{+} \coloneqq (\ket{0}+\ket{1})/\sqrt{2}$, for which $F_C(\ketbra{+} \mid M) = 1/3$. We therefore have $N_C(\ketbra{0} \mid M) = 1/2-1/3=1/6$. This shows that there exists some $G$ and measurement $M$ such that $N_C(\rho \mid M) > 0$ for some $\rho$. Thus demonstrates that $N_C$ identifies strictly more nonclassical states than $N_Q$ over all possible unitary encodings and resource theories.

\subsection{General resource measures based on Fisher information}

The main text focused on the use of Fisher information to identify resource states in general resource theories. Here, we provide addition discussion on how Fisher information based measures can be used to measure the resourcefulness of a quantum state, again, within the context of general resource theories.

A fundamental property of a resource quantifiers is that they should be convex functions of state. If a resource quantifier is not convex, we face a problematic situation where preparing a simple statistical mixture of states, which is a classical process, can increase the amount of quantum resource in the system. In the following proposition we establish that the quantities $N_C$ and $N_Q$ are convex functions, so they avoid this complication.

\begin{proposition} \label{prop::convexity}
 Both $N_C(\rho \mid M)$ and $N_Q (\rho)$ are convex functions of state, i.e.\ $\sum_i p_i N_C(\rho_i \mid M) \geq N_C(\sum_i p_i \rho_i \mid M)$ and $\sum_i p_i N_Q(\rho_i) \geq N_C(\sum_i p_i \rho_i )$, where $p_i \geq 0$ and  $\sum_i p_i =1$.
\end{proposition}

\begin{proof}
	This property follows from the convexity of $F_C$ and $F_Q$~\cite{Tan2019, Sidhu2020}.
	
	We first consider $N_C$. From the definition $N_C (\rho \mid M) \coloneqq F_C(\rho \mid M) - \max_{\sigma \in  \free}F_C(\sigma \mid M)$, we see that the second term $\max_{\sigma \in  \free}F_C(\sigma \mid M)$ does not depend on the input state, so if $F_C(\rho \mid M)$ is convex, so is $N_C(\rho \mid M)$.

	Consider two states $\rho$ and $\sigma$, with statistics $P_\rho(i\mid \theta) = \Tr[\Phi_\theta(\rho)M_i]$  and $P_\sigma(i\mid \theta) = \Tr[\Phi_\theta(\sigma)M_i]$. The convex combination $p\rho + (1-p)\sigma $ leads to the statistics $P_{p\rho+(1-p))\sigma}(i\mid \theta) = pP_\rho(i\mid \theta)+(1-p)P_\sigma(i\mid \theta)$, where $p\in [0,1]$. We see that this is just the convex combination of two distributions $P_\rho(i\mid \theta)$ and $P_\sigma(i\mid \theta)$. It is known that the classical Fisher information $F_C$ is convex function over classical probability distributions so we have $p F_C(\rho \mid M) + (1-p)F_C(\sigma \mid M) \geq F_C(p\rho+(1-p)\sigma \mid M)$. This shows that $F_C(\rho \mid M)$ is convex so $N_C$ is also a convex function of state.
	
	Similarly, $F_Q(\rho)$ is known to be a convex function of state so $N_Q(\rho)$ must also be convex.
	
\end{proof}

In general resource theories, a function of state is considered a resource measure only when it monotonically decreases under some set of ``free" quantum operations $\freeop$. The exact nature of $\freeop$ varies depending on the resource theory under consideration, but they share one thing in common: they always map a classical state to another classical state, i.e.\ if $\Phi \in \freeop$, then $\Phi(\sigma) \in  \free$ for every $\sigma \in  \free$. We will discuss the construction $N^{\max}_C$ based on $N_C$ that naturally satisfies this monotonicity property. Recall that
\begin{align}
N^{\max}_C(\rho) \coloneqq \max_{\Phi_\theta \in \mathcal{P}, M}  N_C(\rho \mid \Phi_\theta , M ),
\end{align} where  $\mathcal{P}$ is the set of all parameter estimation problems satisfying $\max_{\sigma \in  \free} F_C(\sigma \mid \Phi_\theta, M ) \leq 1$. 


The following theorem shows that the above quantity is a resource monotone in general resource theories.

\begin{theorem*}
In any quantum resource theory with resource states $ \free$ and free operations $\freeop$, the quantity $N^{\max}_C(\rho)$ is a resource measure in the sense that it satisfies: (i) $N^{\max}_C(\rho) \geq 0$ for every state $\rho$ and $N^{\max}_C(\rho) > 0$ iff $\rho$ is nonclassical. (ii) $N^{\max}_C$ monotonically decreases under free operations, i.e.\ $N^{\max}_C[\Phi(\rho)] \leq N^{\max}_C(\rho)$ for any $\Phi \in \freeop$. (iii) $N^{\max}_C$ is a convex function of state, i.e.\ $N^{\max}_C[\rho] \leq \sum_i p_i N^{\max}_C(\rho_i)$. 
\end{theorem*}

\begin{proof}
	To prove the first property, we first see that from Theorem~\ref{thm::existence}, there must exist $\Phi_\theta$ and $M$ s.t. $N_C(\rho \mid \phi_\theta, M) > 0 $ if $\rho$ is nonclassical, hence $N_C^{\max}(\rho) > 0$ if $\rho$ is nonclassical. Conversely, if  $N_C^{\max}(\rho) > 0$  then $N_C^{\max}(\rho) > 0$ for some $\Phi_\theta$ and $M$, so the state must be nonclassical. We also see that if $\sigma \in  \free$ is classical, then the identity encoding $\Phi_\theta = \openone$ will reach $N_C^{\max}(\sigma) = 0$ which is the maximum for a classical state.
	 
	Second, we show that $N^{\max}_C(\rho)$ always monotonically decreases under quantum operations. Suppose $\Phi' \in \freeop$ is a free operation, and let $\Phi_\theta^*,M^*$ be optimal for the state $\Phi'(\rho)$ such that $N_C^{\max}[\Phi'(\rho)] = N_C[\Phi'(\rho) \mid \Phi_\theta^*,M^*]$. We have the following chain of inequalities:
	\begin{align}
	&N^{\max}_C(\rho) \\
	&= \max_{\Phi_\theta \in \mathcal{P}, M} \left[ F_C(\rho \mid \Phi_\theta, M )- \max_{\sigma \in  \free} F_C(\sigma \mid \Phi_\theta, M )\right ]\\
	&\geq F_C(\rho \mid \Phi^*_\theta\circ \Phi', M^* )- \max_{\sigma \in  \free} F_C(\sigma \mid \Phi^*_\theta\circ \Phi', M^* ) \\
	& = F_C[\Phi'(\rho) \mid \Phi^*_\theta, M^* ]- \max_{\sigma \in  \free} F_C[\Phi'(\sigma) \mid \Phi^*_\theta, M^* ] \\
	& \geq F_C[\Phi'(\rho) \mid \Phi^*_\theta, M^* ]- \max_{\sigma \in  \free} F_C[\sigma \mid \Phi^*_\theta, M^* ] \\
	&=N_C^{\max}[\Phi'(\rho)].
	\end{align} The first inequality above comes from the fact that  $\Phi_\theta^* \circ \Phi',M^*$ is in general suboptimal for the state $\rho$. The second inequality comes from the fact that $\Phi' \in \freeop$ is a free operation, so $\Phi'(\sigma) \in  \free$ for every $\sigma \in  \free$.
	
	Finally, we show that it is a convex function of state. Suppose $\rho = \sum_i p_i \rho_i$, and let $\Phi_\theta^*, M^*$ and $\Phi_{\theta,i}^*, M_i^*$ be the optimal channel encodings and measurements for $\rho$ and $\rho_i$ respectively.  We then have
	\begin{align}
	N^{\max}_C(\rho) &= N_C(\rho \mid \Phi_\theta^*, M^*) \\
	&\leq \sum_i p_i N_C(\rho_i \mid \Phi_\theta^*, M^*)\\
	&\leq \sum_i p_i N_C(\rho_i \mid \Phi_{\theta,i}^*, M_i^*)\\
	&= \sum_i p_i N^{\max}_C(\rho_i),
	\end{align} where the first inequality comes from the convexity of $N_C$ (Proposition~\ref{prop::convexity}) and the second inequality comes from the suboptimality of $\Phi_\theta^*, M^*$ for the state $\rho_i$.
	
\end{proof}

We have shown that $N_C^{\max}$ is a resource monotone. One may also consider the quantum Fisher information based measure 
\begin{align}
N^{\max}_Q(\rho) \coloneqq \max_{\Phi_\theta \in \mathcal{P}}  N_Q(\rho \mid \Phi_\theta  ),
\end{align} where  $\mathcal{P}$ is the set of all parameter estimation problems satisfying $\max_{\sigma \in  \free} F_Q(\sigma \mid \Phi_\theta) \leq 1$. It can also be shown that $N^{\max}_Q$ also monotonically decreases under free operations, following largely identical arguments as $N_C^{\max}$.

\end{document}